\documentclass[journal]{IEEEtran}

\usepackage{mathrsfs}
\usepackage{amsfonts}
\usepackage{amsthm}
\theoremstyle{remark}
\newtheorem{definition}{Definition}

\newtheorem{theorem}{Theorem}
\newtheorem{lemma}{Lemma}
\newtheorem{remark}{Remark}
\usepackage{algorithm}
\usepackage{hyperref}
\ifCLASSINFOpdf
   \usepackage[pdftex]{graphicx}
\else
   \usepackage[dvips]{graphicx}
   \graphicspath{{../eps/}}
   \DeclareGraphicsExtensions{.eps}
\fi

\usepackage[cmex10]{amsmath}
\usepackage{algorithmic}

\begin{document}
\title{How Much Frequency Can Be Reused in 5G Cellular Networks---A Matrix Graph Approach}

\author{Yaoqing~Yang,~\IEEEmembership{Student~Member,~IEEE,}
        Bo~Bai,~\IEEEmembership{Member,~IEEE,}
        and~Wei~Chen,~\IEEEmembership{Senior~Member,~IEEE}\vspace{-0.7cm}% <-this % stops a space
\thanks{Yaoqing Yang is with the Department
of Electrical and Computer Engineering, Carnegie Mellon University, Pittsburgh,
PA, 15232 USA e-mail: yyaoqing@andrew.cmu.edu.}% <-this % stops a space
\thanks{Bo Bai and Wei Chen are with the Department of Electronic Engineering, Tsinghua University, Beijing, 100084 China e-mail: {eebobai,wchen}@tsinghua.edu.cn.}% <-this % stops a space
%\thanks{Manuscript received December 1, 2013; revised ***.}
}


\markboth{}%
{}



\maketitle

\begin{abstract}
%\boldmath
The 5th Generation cellular network may have the key feature of smaller cell size and denser resource employment, resulted from diminishing resource and increasing communication demands. However, small cell may result in high interference between cells. Moreover, the random geographic patterns of small cell networks make them hard to analyze, at least excluding schemes in the well-accepted hexagonal grid model. In this paper, a new model---the matrix graph is proposed which takes advantage of the small cell size and high inter-cell interference to reduce computation complexity. This model can simulate real world networks accurately and offers convenience in frequency allocation problems which are usually NP-complete. An algorithm dealing with this model is also given, which asymptotically achieves the theoretical limit of frequency allocation, and has a complexity which decreases with cell size and grows linearly with the network size. This new model is specifically proposed to characterize the next-generation cellular networks.\\
\end{abstract}

% \fi
\IEEEpeerreviewmaketitle

\section{Introduction}
\IEEEPARstart{F}{requency} reuse and the cellular concept~\cite{1} is the driven force behind several decades of innovations in the wireless communication field. Many pioneering works~\cite{2}\cite{3}\cite{4} are based on the convenient assumptions that cells, frequency reuse patterns and even user demands, are geographically periodic, often simulated by a regular hexagonal grid model. However, these assumptions have long been suspicious by research simulations and industry practices~\cite{1}\cite{5}. A recent survey~\cite{7} suggested that the existing 4G and future mobile networks may actually have a geographic pattern which falls in between the regular grid model and a totally random graph. This phenomenon indicates the necessity of a new frequency reuse approach.\\
\par
Moreover, the recently emerging 5th generation cellular networks, which are believed to have small cells~\cite{7}\cite{8}, also poses new challenges to the conventional approaches. Small cells are advantageous in higher energy efficiency and indoor coverage, but may cause higher inter-cell interference. This means that the classic one-base-station-downlink model~\cite{9} cannot be used here. Many emerging interference management techniques are devoted to this problem, e.g., Fractional Frequency Reuse~\cite{6}, multi-cell coordination~\cite{10} and cooperations~\cite{11}. Therefore, we would like to answer the question: how much frequency can be reused in a highly cooperative and high-interference small-cell network?\\
\par
In this paper, we propose a new approach called the \emph{matrix graph} to answer this question. A matrix graph is a lattice-like conflict graph while each lattice point is substituted by a small random graph called a cell. The vertices in the graph represent communication links~\cite{17}~\cite{19}, i.e. either uplink or downlink, while the edges represent interference. Confliction graph is widely adopted in cellular communications~\cite{13}-\cite{20} and frequency allocation in a conflict graph can be conveniently treated as multi-coloring problems. We still consider coloring, i.e., frequency allocation, in matrix graphs. But we will show why this lattice-like matrix graph is especially suitable to deal with frequency allocation in the 5G network.\\
\par
In the matrix graph model, we make the cell shapes and sizes random, but still reserve a lattice pattern. As stated above, this matches the real 4G cellular network structure shown in~\cite{7}. Thus, the first merit of the matrix graph model is its high resemblance to real-world networks. The second advantage of the matrix graph is its computation efficiency in a high-interference small-cell network. As shown in Section~\ref{Analysis}, if we increase the inter-cell interference, the computation complexity of multi-coloring will be lower. Preceding works widely recognized the trend of small cells, but seldom did they actively design network models and algorithms to meet this trend.(might need citations)\\
\par
The third virtue of the matrix graph, compared to other graph-based models~\cite{13}-\cite{18}, is being tractable to reach the fundamental limit of frequency allocation. Although in this paper, obtaining the optimal frequency allocation in a matrix graph is proved to be NP-complete, we still obtained a linear-time approximation algorithm with a solution within a bounded gap to the optimal value. This means that for the small-cell network, we can directly tell how much frequency can be reused after the corresponding matrix graph has been constructed. This is in contrast with frequency allocation in general graphs. In fact, frequency allocation problems, like multi-coloring and the related Maximum Weighted Independent Set (MWIS) problems are MAX SNP problems~\cite{20}, which means that even making a performance-guaranteed approximation is NP Hard (See Section~\ref{Graph_Partition_Complexity}). That is why previous works on coloring-based heuristics often lack analytical results. Moreover, it is showed in this paper that coloring a one-dimensional matrix graph has linear-time solution. This model itself is also important, because large networks can be one-dimensional, e.g., a femtocell network along a long road or a wifi-network in a long train.\\
\par
In this paper, the final goal is to achieve the best reuse-interference tradeoff, i.e., obtaining the maximum frequency bands used by each communication link without interference. This is often called the \emph{Maximum Service Frequency Allocation} (MSFA) and has been widely accepted as a benchmark of efficiency, e.g., see survey~\cite{18}. Our method does not rely on specific resource type. For simplicity we assume resources to be frequency bands or OFDMA subcarriers. The only requirement is that any two resources are orthogonal and a specific resource cannot be reused by interfering communication links. There are both works on assuming links~\cite{17}\cite{19} or User Terminals~\cite{15}\cite{16} as confliction agents. We follow the first one because in a 5G network, there may be cooperations between cells and thus, one UT may have a few communication links. Also, we assume that all the heterogeneous base stations are linked to the central network with wired backhaul~\cite{7}\cite{10}\cite{11}. Thus, scheduling can be carried out in the whole network. This large-scale scheduling only incurs an $\mathcal{O}(MN)$ overhead where the network size is $M$-by-$N$. So it prevails exact algorithms which usually have exponential complexity.\\

In summary, we designed a graph-based approach suitable in small cells, which complements the insufficiencies of hexagonal grid models and conflict graph models; we designed an algorithm to allocate frequencies efficiently with a computation complexity growing linearly with network size. The paper is arranged as follows: in Section~\ref{modeling}, the system model based on stochastic geometry is covered; in Section~\ref{Graph_Partition_Section}, the matrix graph approach is formulated; in Section~\ref{Coloring_Algorithm}, a high-efficiency and low-complexity scheduling algorithm is proposed and analysed; Section~\ref{Simulations} discusses simulation results.
\section{System Model}\label{modeling}
\vspace{0.2cm}
We use a graph $G=(V,E)$ (as shown in Fig.~\ref{m_graph}(a)) to represent the network. Each vertex $v\in V$ denotes a communication link in the network and each edge $e\in E$ indicates a confliction between two neighboring communication links. We consider communication links, rather than User Terminals (UTs), as the conflicting agents~\cite{17}\cite{19}. The term \emph{vertex} and \emph{communication link} will be used interchangeably. The graph $G=(V,E)$ is generated by the random connection model~\cite{25}\cite{26} on a rectangular area. The vertices in $V$ is given by a homogeneous Poisson point process (PPP) with point density $\lambda$. For each pair of nodes $v_i$ and $v_j$, $(v_i,v_j)\in E$ with probability $g(\mathbf{x}_i-\mathbf{x}_j)$ where $\mathbf{x}_i$ and $\mathbf{x}_j$ respectively means the position of $v_i$ and $v_j$ in $\mathbb{R}^2$ and $g$ is a function from $\mathbb{R}^2$ to $[0,1]$. $g(\mathbf{x})$ should satisfy $\lim_{|\mathbf{x}|\rightarrow\infty}g(\mathbf{x})=0$ and $e(g):=\int_{\mathbf{x}\in \mathbb{R}^2} {g(\mathbf{x})d\mathbf{x}}<\infty$. In fact, $\lambda e(g)$ is the expected connections per node. Apart from the previous two conditions, we make a further assumption
\begin{equation}\label{RanCon_As1}
  g(\mathbf{x}_1)<g(\mathbf{x}_2),\text{ if }|\mathbf{x}_1|>|\mathbf{x}_2|,
\end{equation}
which means that interference probability decreases when distance increases.

Assume there are $C$ colors \emph{colors} $\Lambda=\{1,2,...,C\}$ to allocate. Each color $c$ can be treated as a frequency band. A mapping $\mathscr{C}$, called \emph{multi-coloring}, can be used to represent frequency allocation.
\begin{equation}\label{eq0}
\mathscr{C}(v,c) = \left\{ {\begin{array}{*{20}{c}}
{1,}\\
{0,}
\end{array}\begin{array}{*{20}{c}}
\text{when $c$ is allocated to $v$,}\\
\text{otherwise.}
\end{array}} \right.
\end{equation}
We write the \emph{reuse ratio}
\begin{equation}\label{2}
{{{f}}_v} =\frac{1}{C}\mathop \sum \limits_{c  = 1}^{C} {\mathscr{C}(v,c)\mu(v,c)}
\end{equation}
for the portion of the total frequency bands used by $v$. $\mu(v,c)$ stands for the \emph{color weight} which will be discussed in Remark~\ref{remark_1}. Then the conventional frequency allocation problem can be written as a multi-coloring problem:
\begin{equation}\label{color}
\begin{split}
&\max_{\mathscr{C}} {\;}\bar f =({\mathop \sum \limits_{v\in V}  w_{v}f_{v}})/({\mathop \sum \limits_{v\in V}  w_{v}}),\\
&\text{s.t.}{\;}\mathscr{C}(v_1,c)+\mathscr{C}(v_2,c)\le 1, \text{ if}{\;}(v_1,v_2)\in E.
\end{split}
\end{equation}
where $w_v$ denotes the \emph{vertex weight} which will be explained in Remark~\ref{remark_1}. The constraint in (3) ensures that no conflicting links are assigned the same frequency band. Note that multi-coloring allows each vertex to be assigned more than 1 colors, in order to achieve a higher reuse ratio~\cite{14}\cite{15}\cite{16}\cite{19}. However, this problem is NP-complete~\cite{18}. The next subsection shows how to simplify the problem with a matrix graph approach.
\begin{remark}\label{remark_1}
A quick example about vertex weight $w_v$ is the frequency allocation between cell centers and cell edges~\cite{6}. We will assign bigger weights to cell edges, where channel conditions are poor. For color weights $\mu(v,c)$, if we make the assumption that
\[\mu(v,c)=\log(1+SNR_v)=\log(1+P(v,c)h(v,c)/\sigma^2),\]
where $P(v,c)$ and $h(v,c)$ denotes the transmit power and channel gain, then optimizing $f_v$ becomes optimizing channel capacity. We only consider large-scale fading so that this information is available at the central node. We can also assume $\mu(v,c)=1$. Then ${{f}}_v\in[0,1]$ is exactly the ratio of the available frequency bands that $v$ can utilize. Optimizing $f_v$ now is the same as the Maximum Service Frequency Allocation~\cite{18}. Therefore, optimizing $f_v$ is consistent with multiple classic optimization problems in communications.
\end{remark}
\section{A Matrix Graph Approach}\label{Graph_Partition_Section}
\begin{figure}
\centering
\includegraphics[scale=0.35]{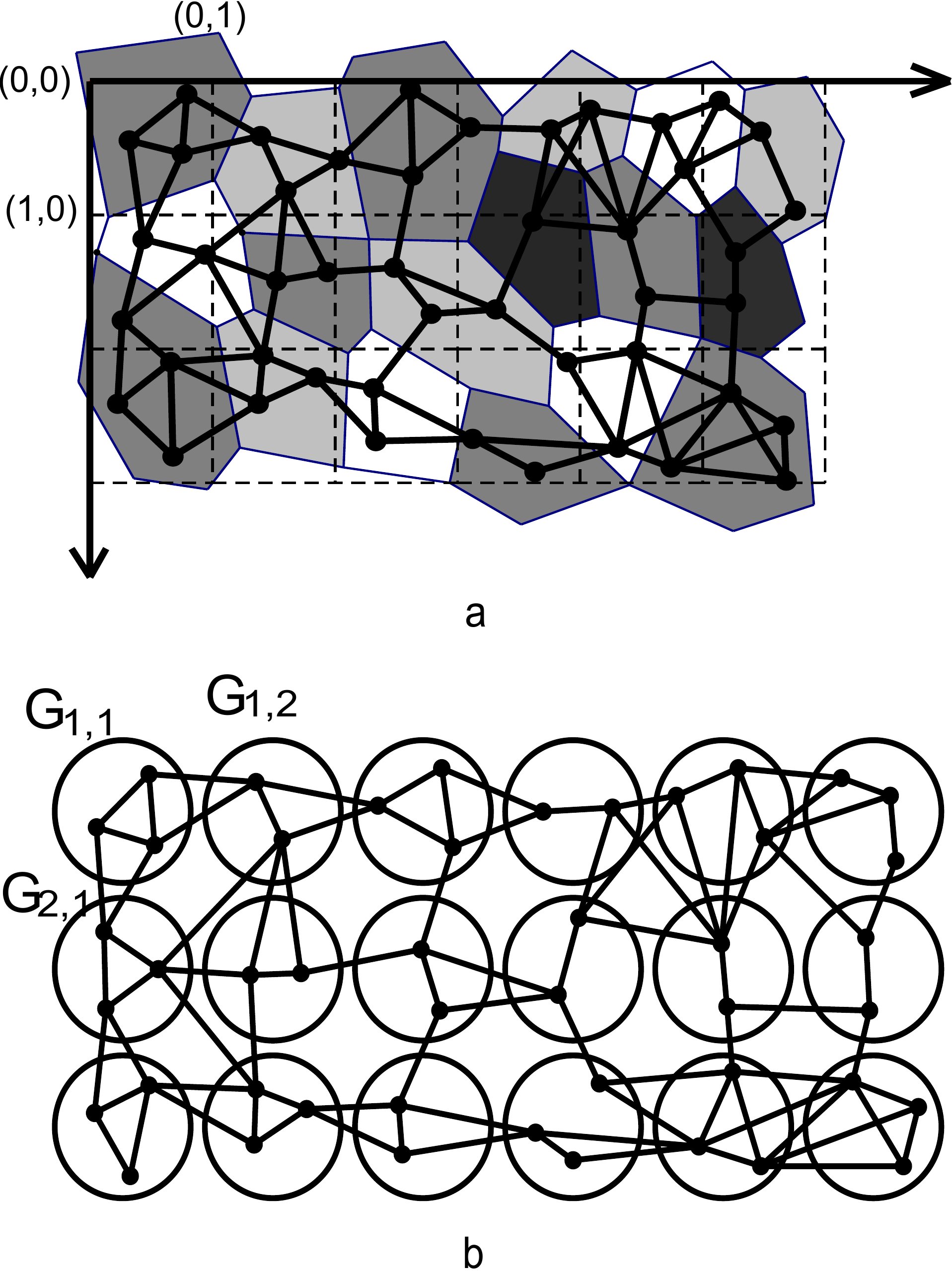}
\caption{System Model: (a). A conflict graph. (b). The corresponding matrix
graph. Only neighboring cells interfere with each other.}\label{fig1，2}\vspace{-0.5cm}
\label{m_graph}
\end{figure}
In this section we show that in small-cell networks, a confliction graph can be transformed into a matrix graph with bounded performance loss in multi-coloring. Fig.~\ref{m_graph}(a) shows a typical small-cell conflict graph. Each cell has $1\sim3$ communication links. However, inter-cell and intra-cell interference is quite complicated. This motivates us to transform this kind of random graph into a more manageable structure--the matrix graph (as shown in Fig.~\ref{m_graph}(b)).
\subsection{Matrix Graph Formulation}
We partition the rectangular area with $M+1$ parallel lines horizontally and $N+1$ parallel lines vertically with distance $a$ which will be determined later (as shown Fig.~\ref{m_graph}(a)). Therefore, the rectangular area is partitioned into $MN$ squares. Meanwhile, vertices are separated into each square
\begin{equation}\label{3}
{ {V}} = \mathop \bigcup \limits_{m=1}^{M}\mathop \bigcup \limits_{n = 1 }^{N}  {{ {V}}_{m,n}},
\end{equation}
where $V_{m,n}$ denotes the vertices in the square constituted by lattice points $\{(m-1,n-1),(m,n-1),(m-1,n),(m,n)\}$ and
\begin{equation}\label{4}
V_{m,n}=\{v_{m,n}^i\}_{i=1}^{l_{m,n}}.
\end{equation}
Here $l_{m,n}$ is the number of vertices in this square and $v_{m,n}^i$ means the $i$th vertex. We use \(G_{m,n}\) to denote the induced graph by $V_{m,n}$ from $G$, which means that $G_{m,n}=(V_{m,n},E_{m,n})$ and $E_{m,n}=\{e\in E|e=(v_{m,n}^i,v_{m,n}^j),1\le i<j\le l_{m,n}\}$. Each \(G_{m,n}\) is surrounded by a circle in Fig. 2(b). From now on, \(G_{m,n}\) will also be called a \emph{cell}.\\

Then we induce all edges from $G$ to construct a matrix graph like Fig.~\ref{m_graph}(b), except those edges that connect non-adjacent cells in the constructed graph, e.g., \(G_{m,n}\) and \(G_{m+2,n}\). The constructed matrix graph $\tilde{G}=(V,\tilde{E})$ satisfies
\[\tilde{E}=E-\{(v_{m,n}^i,v_{p,q}^j)|\text{ and }|m-p|>1\text{ or }|n-q|>1\}.\]
Suppose we now carry out the multi-coloring algorithm and find the best coloring in the constructed matrix graph, then it may not be a legal coloring in the previous graph because we may have neglected the conflictions of some dropped edges in the graph partitioning. So, if the two vertices on a dropped edge $(v_i,v_j)\in E$ assigned the same color $c$, we will cancel the usage of $c$ in one of these two vertices. We call this procedure the \emph{validity check}.
\begin{lemma}\label{SG_lmm}
Denote the optimal reuse ratio in the original random graph by $\bar{f}^{*}$ and the reuse ratio after doing graph partitioning, optimal multi-coloring and validity check by $\bar{f}_{m}$. Then
\begin{equation}\label{Drop_Edge_optimal}
E[\bar{f}^{*}-\bar{f}_{m}]<\frac{1}{2}\lambda \int_{\mathbf{x}\in\Omega}g(\mathbf{x})d\mathbf{x},
\end{equation}
where $g(\cdot)$ is the connection probability function in~\eqref{RanCon_As1} for the random connection model and $\Omega$ is the region
\[\Omega=\{\mathbf{x}=(x^1,x^2)\in \mathbb{R}^2||x^1|>a\text{ or }|x^2|>a\},\]
where $a$ is the distance between parallel lines in partitioning.
\end{lemma}
Lemma~\ref{SG_lmm} is proved in the Appendix~\ref{Graph_Partition_performance_loss}. This lemma ensures that the performance loss of the graph partitioning procedure is bounded, and the loss is especially small when $g(\cdot)$ is decreasing quickly. For example, if we consider the boolean model~\cite{26} where
\[g(\mathbf{x})=\mathbb{I}_{\{|\mathbf{x}|\le 2r\}},\]
then the performance loss is zero if the chosen parameter $a>r$. This lemma is the justification for matrix graphs and hence we can give out the definition of the matrix graph.
\begin{definition}
A matrix graph is a conflict graph $G=(V,E)$ where $V$ satisfies~\eqref{3}\eqref{4} and an edge \((v_{m_1,n_1}^i,v_{m_2,n_2}^j)\in E\) only if
\begin{equation}\label{5}
(|m_1-m_2|\le 1) \text{ and } ( |n_1-n_2|\le 1).
\end{equation}
The constraint~\eqref{5} ensures that only neighboring cells in the matrix graph have conflictions.\\
\end{definition}
By abuse of notation in~\eqref{2}, we use $f_{m,n}^i$ to represent the reuse ratio for $v_{m,n}^i$ in a matrix graph, meaning that
\begin{equation}\label{6}
{{{f}}_{m,n}^i}={f_{v_{m,n}^i}}=\frac{1}{{ C}}\mathop \sum \limits_{c  = 1}^{C} {\mathscr{C}(v_{m,n}^i,c )\mu(v_{m,n}^i,c )},
\end{equation}
where $\mathbf{\mu}=(\mu(v_{m,n}^i,c))$ denotes the color weight discussed in Remark~\ref{remark_1}. Then the ultimate goal of maximizing the \emph{weighted reuse ratio} in a matrix graph can be written as
\begin{equation}\label{7}
\begin{split}
&\max_{\mathscr{C}}{\;}{\;}\bar f = \frac{\mathop \sum \limits_{m,n=  1}^{M,N} \mathop \sum \limits_{i = 1}^{{l_{m,n}}} w_{m,n}^if_{m,n}^i}{\mathop \sum \limits_{m ,n=  1}^{M,N} \mathop \sum \limits_{i = 1}^{{l_{m,n}}} w_{m,n}^i},\\
&\text{s.t.}{\;}\mathscr{C}(v_1,c)+\mathscr{C}(v_2,c)\le 1, \text{ if}{\;}(v_1,v_2)\in E.
\end{split}
\end{equation}
where $w_{m,n}^i$ indicates the vertex weight discussed in Remark~\ref{remark_1}. This is called the \emph{matrix graph coloring} problem (MGC).
\subsection{Multi-coloring Complexity in Matrix Graphs}\label{Graph_Partition_Complexity}
In this section we briefly discuss the computation advantage of matrix graphs over general graphs. By accepting the loss bounded by (7), we expect to gain advantage in multi-coloring computation. Although we have the following theorem which will be proved in the Appendix~\ref{Proving_NP_complete}, we still obtained an linear complexity approximate algorithm to achieve a bounded performance.
\begin{theorem}
MGC problem is NP-complete.
\end{theorem}

In fact, optimal multi-coloring problem is polynomially equivalent to the maximum weighted independent set problem(MWIS), which has been proved to be NP-complete. Moreover, MWIS problem on a general graph is proved to be in the complexity equivalent class MAX SNP problems~\cite{20}, so there is an $\varepsilon>0$ such that the MWIS problem cannot be approximated in polynomial time with performance ratio greater than $\frac{1}{n^\varepsilon}$, unless P=NP. Therefore, the matrix graph multi-coloring is much better because we can use linear time to construct a solution with only a bounded performance gap to the optimum (as shown in Section~\ref{Analysis}).\\

A finer result about bounded-degree graph is that MWIS problem on general bounded-degree graph is APX but APX-complete~\cite{21}. That is to say, MWIS performance on bounded-degree graphs can be approximated within some constant (e.g., $50\%$) but it cannot be approximated arbitrarily close to $100\%$ with polynomial-time algorithms, unless P=NP. However, in a bounded-degree matrix graph, we will show in Section~\ref{Analysis} that for any $\epsilon>0$, we can choose a parameter $L=\min(M,\frac{1}{\epsilon})$ such that an algorithm with complexity exponential in $L$ achieves $1-\epsilon$ of the optimum. Thus, the bounded-degree matrix graph multi-coloring can be divided into the class PTAS~\cite{21}, instead of APX-complete.\\

In summary, matrix graphs can be treated as approximated models specified for small-cell networks, which are much easier for frequency allocation than general conflict graph.
%Moreover, even if the maximum degree of a matrix graph grows logarithmically with the network size $M$, $K$ is still growing polynomially with $M$. Therefore, the complexity $\mathcal(CK^{L-1}MN)$ is still polynomial for any fixed $L$, which means that multi-coloring in matrix graphs with logarithmic node degree is still in PTAS.\\
%Because in this case, the maximum number of independent sets in each cell, or the parameter $K$, is a constant, so for any fixed $L$, the complexity $\mathcal(CK^{L-1}MN)$ is linear in the problem size. But $L$ can goes to infinity (in fact when $L$ goes to $M$, the MWIS problem is already solved exactly), which means that the performance can be arbitrarily close to $100\%$.
\section{Solving The Matrix Graph Coloring Problem}\label{Coloring_Algorithm}
\vspace{0.2cm}
In this section, we first use a scheme called \emph{floor division} to map the original MGC problem into many one-dimensional Maximum Weighted Independent Set (MWIS) problems. Then we solve each MWIS problem and combine the results with approximation techniques. The final algorithm to solve the MGC problem is outlined in Algorithm 1. In subsection~\ref{Describing_Algorithm} we give an overview of Algorithm 1. In subsection~\ref{Analysis}, we analyze the performance and complexity of Algorithm 1.
\subsection{Approximation Algorithm with a Floor Dividing Method}\label{Describing_Algorithm}
\vspace{0.2cm}
\subsubsection{Finding Independent Sets in one-dimensional graph}
In a matrix graph, an IS generally represents a subset of communication links who do not conflict with each other when utilizing the same frequency. Specifically, for a graph \(G=(V,E)\), a vertex subset \(S\subset V\) is called independent if no two vertices in $S$ share the same edge in $E$. An IS $S$ in a matrix graph $G$ can be decomposed into $MN$ small ISs
\begin{equation}\label{8}
{ {S}} = \mathop \bigcup \limits_{m=1}^{M}\mathop \bigcup \limits_{n = 1 }^{N}  {{ {S}}_{m,n}},
\end{equation}
and each \(S_{m,n}\subset V_{m,n}\) is an independent set in the cell $G_{m,n}$. For each vertex \(v_{m,n}^{i}\), if we use \(q_{m,n}^{i}\in\{0,1\}\) to denote whether \(v_{m,n}^{i}\in S_{m,n}\), we can define the \emph{normalized weighted cardinality} (NWC) $|\cdot|_N$ of \(S\) as
\begin{equation}\label{9}
\left| { {S}} \right|_N =\frac{\mathop \sum \limits_{v \in V}q_v u_v}{\mathop \sum \limits_{v \in V}u_v} =\frac{\mathop \sum \limits_{m,n = 1}^{M,N} \sum \limits_{i=1}^{l_{m,n}} q_{m,n}^i u_{m,n}^i}{\mathop \sum \limits_{m,n = 1}^{M,N} \sum \limits_{i=1}^{l_{m,n}} u_{m,n}^i},
\end{equation}
where $u_{m,n}^i$ is the vertex weight. If $u_{m,n}^i=1$ for all vertices, the NWC simply equals to ratio of $|S|/|V|$ where $|\cdot|$ means cardinality. It is clear that $|S|_N$ takes value in $[0,1]$. The indicator vector \(\mathbf{q}=(q_{m,n}^{i})\) in~\eqref{9} can represent the solution $S$. In the following we call this $\mathbf{q}$ the \emph{indicator representation} of an independent set.\\

\begin{definition}We call $S^*\subset V$ the \emph{maximum weighted independent set} (MWIS) of graph $G=(V,E)$ if it is an independent set with the maximum normalized weighted cardinality~\eqref{9}.\\
\end{definition}
\begin{lemma}\label{lmm_DP}
Finding MWIS in a one-dimensional matrix graph can be completely solved with $\mathcal{O}(KN)$ time complexity by dynamic programming~\cite{12}, where $K$ is the supremum of the number of Independent Sets in each cell $G_{m,n}$.
\end{lemma}
\begin{proof}
See Appendix \ref{Dynamic_Programming}.
\end{proof}
A one-dimensional matrix graph is a matrix graph with height $M = 1$. Lemma~\ref{lmm_DP} ensures the linear complexity of finding the MWIS on a one-dimensional matrix graph. Therefore, the original two-dimensional matrix graph can be partitioned into many one-dimensional subgraphs and MWISs can be found on each of them efficiently. Section~\ref{Analysis} shows that this scheme can achieve a bounded optimality gap.\\

\subsubsection{Floor Dividing and the matrix graph decomposition}
\begin{figure}
  \centering
  \includegraphics[scale=0.25]{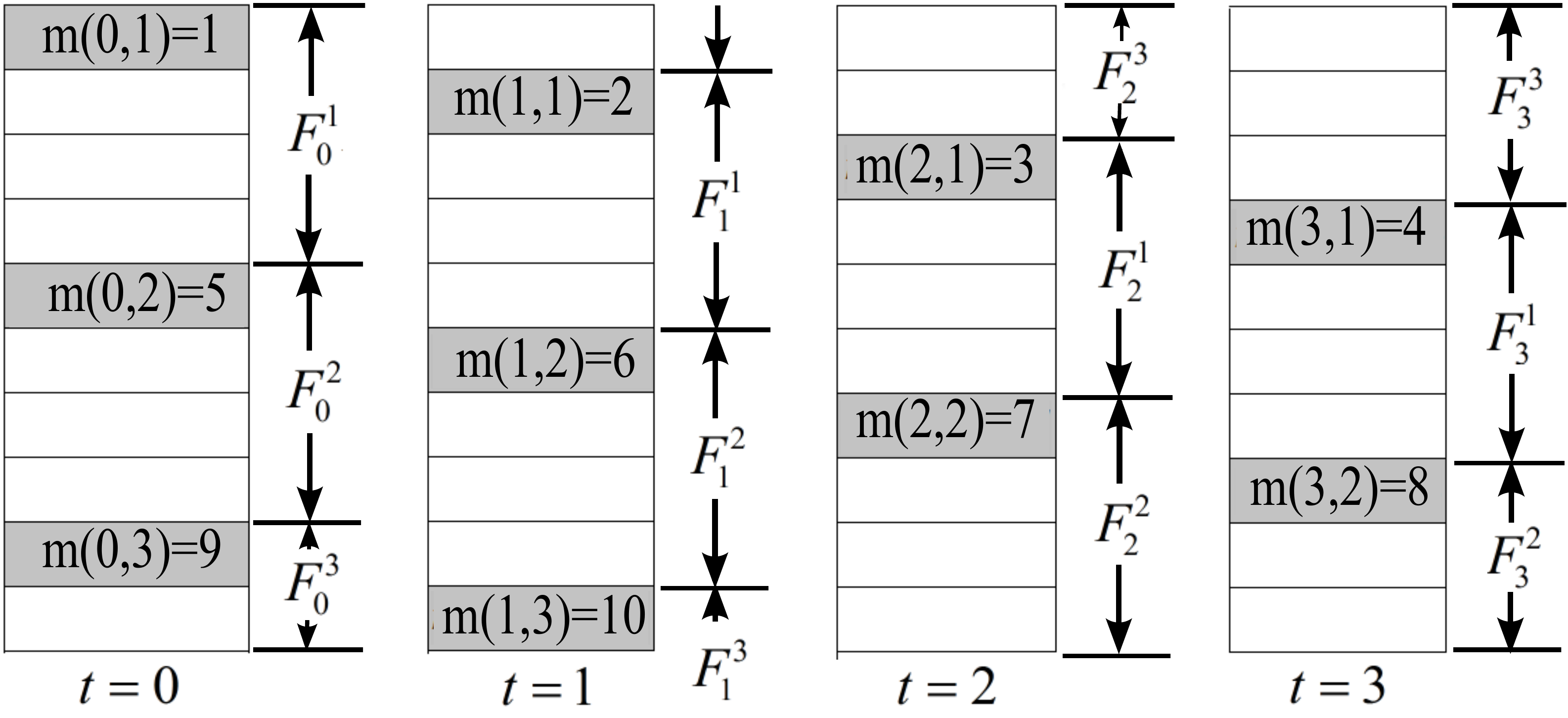}\\
  \caption{\emph{A floor division scheme when $M = 10$, $L = 4$, and $Q = 3$, with $L=4$ floor divisions. Marginal rows are colored differently.}\vspace{-0.5cm}}\label{fig2}
\end{figure}
In order to decompose the MWIS problem, we need to divide the whole M-by-N matrix graph into many one-dimensional subgraphs. Thus, a method called Floor Dividing is proposed. This scheme concurrently separates several copies of the $M$-by-$N$ graph into several slender subgraphs (as shown in Fig.~\ref{fig2}) and views each subgraph as a one-dimensional matrix graph. First we choose a positive integer $L<M$ as a parameter, called the \emph{floor height}. We divide $M$ by $L$ and get
\begin{equation}\label{10}
M=L(Q-1)+r,0 < r \le L.
\end{equation}
It is notable that this division rounds up to get the quotient $Q$. Then we divide the row set $F=\{1,2,...,M\}$ of $G$ into $Q$ subsets \(F=\mathop \cup \limits_{j =  1 }^{Q}  {F_t^j}\), which represents one way of dividing the matrix graph into $Q$ slender layers. We call each subset $F_t^j$ a \emph{floor} and call this set division the $t$th \emph{floor division}. For example, for $t=0$,
\begin{equation}\label{11}
\begin{split}
&F_0^j=\{L(j-1)+1,L(j-1)+2,...,Lj\},j=1,2,...,Q-1,\\
&F_0^Q=\{L(Q-1)+1,L(Q-1)+2,...,M\}.
\end{split}
\end{equation}
This division is like dividing a mansion of height $M$ into $Q$ floors. In a floor division, the first $Q-1$ floors have $L$ rows while the last one has $r\le L$ rows. Fig.~\ref{fig2} shows 4 floor divisions. In each floor, one row might be defined as a \emph{marginal row}, so that if all marginal rows in one floor division are eliminated, the remaining rows in different floors become non-adjacent. Thus, the MWIS can be found in all non-marginal rows by searching for the MWIS in each floor excluding the marginal row. A \emph{floor division scheme} is a group of different floor divisions. The following lemma ensures the existence of a floor division scheme that makes each row being the marginal row exactly once.\\
\begin{lemma}\label{fl_div_lmm}
For a given $M$-by-$N$ matrix graph and a floor height $L<M$, a floor division scheme with $L$ floor divisions can be constructed, with the $t$th division written as $\{F_t^j\}_{j=1}^Q$. Each $F_t^j$ contains at most one marginal row $m(t,j)$, s.t.\\
(i) Each division $t$ divides $G$ into $Q$ subgraphs which are only adjacent on marginal rows;\\
(ii) All marginal rows constitute $F=\{1,2,...,M\}$.
\end{lemma}
\begin{proof}
See Appendix C.
\end{proof}
The main idea is shown in Fig.~\ref{fig2}. A cyclic construction scheme is utilized to make the required floor division scheme. Lemma~\ref{fl_div_lmm} suggests that the entire graph $G$ can be divided into $Q$ subgraphs in $L$ different ways. Since each subgraph has a height bounded by $L$, we can view each one as a one-dimensional matrix graph and completely solve it with linear complexity according to results of Lemma~\ref{lmm_DP}. The aim of concurrently dividing $L$ copies is to ensure the property 2 of Lemma~\ref{fl_div_lmm} and Theorem 2.\\
\subsubsection{MGC Algorithm 1}
Assume a matrix graph $G=(V,E)$ is going to be colored with a color pool \(\Lambda=\{1,2,...,C\}\). The way to solve it is to assign each color $c$ to an MWIS. In order to find the MWIS in an $M$-by-$N$ graph, the floor dividing can be utilized to cut $L$ copies of $G$ into subgraphs and on each subgraph the MWIS can be obtained. Based on this idea, Algorithm 1 is given. Its performance is guaranteed by Theorem 2 in the next subsection. In the algorithm, the approximation scheme is that instead of searching for MWIS in a whole subgraph, we find MWIS of each subgraph excluding the marginal row. Since marginal rows are between non-adjacent MWISs and contains no vertices in the MWISs, additional vertices can be added in to stuff the marginal rows. A larger floor height $L$ can result in a higher complexity, but a more accurate approximation.
\begin{algorithm}\caption{Solving MGC problem}\label{alg:B}
\textbf{Input}: A matrix graph $G=(V,E)$, a color pool $\Lambda$, vertex weight $\mathbf{w}$ and color weight $\mathbf{\mu}$\\
\textbf{Output}: A matrix graph Coloring \(\mathscr{C}=(\mathscr{C}(v_1,c))\) which optimizes \(\bar{f}\) in ~\eqref{7} to $1-1/L$ of the optimal value.\\
\\
Initialize\\
\\
\\\textbf{/*Floor Dividing*/}\\
Calculate the floor dividing scheme $F_t^j$, $\forall t\in\{0,1,...,$ $L-1\},\forall j\in\{1,2,...,Q\}$ based on Lemma~\ref{fl_div_lmm};\\
\\
\\\textbf{/*MWIS for each color*/}\\
\textbf{FOR} each color $c\in \Lambda$\\
${\;\;\;\;}$Solve a MWIS problem in $G$ associated with vertex\\
${\;\;\;\;}$weights \(\mathbf{u}=(u_{m,n}^{i})\) defined as
\begin{equation}\label{12}
u_{m,n}^{i}=w_{m,n}^i \mu(v_{m,n}^i,c),\forall m,n,i
\end{equation}

\({\;\;\;\;\;\;}\)\textbf{FOR} each floor division \(t\) from \(0\) to \(L-1\)\\
\\
\({\;\;\;\;\;\;}\)\({\;\;\;\;\;\;}\)\textbf{/*MWIS for each one-dimensional graph*/}\\
\({\;\;\;\;\;\;}\)\({\;\;\;\;\;\;}\)\textbf{FOR} each floor $j \in \{1,2,...,Q\}$\\
\({\;\;\;\;\;\;}\)\({\;\;\;\;\;\;}\)\({\;\;\;\;\;\;}\)Set $\bar{F}_t^j=F_t^j\setminus$the marginal row;\\
\({\;\;\;\;\;\;}\)\({\;\;\;\;\;\;}\)\({\;\;\;\;\;\;}\)View all rows that have index $m\in \bar{F}_t^j$ as a
\({\;\;\;\;\;\;}\)\({\;\;\;\;\;\;}\)\({\;\;\;\;\;\;}\)one-dimensional matrix graph $\bar{G}_t^j$;\\
\({\;\;\;\;\;\;}\)\({\;\;\;\;\;\;}\)\({\;\;\;\;\;\;}\)Use Algorithm 2 to find a MWIS $\bar{S}_t^j$ in one-\\
\({\;\;\;\;\;\;}\)\({\;\;\;\;\;\;}\)\({\;\;\;\;\;\;}\)dimensional Graph $\bar{G}_t^j$ with no extra constraints;\\
\({\;\;\;\;\;\;}\)\({\;\;\;\;\;\;}\)\textbf{END}\\
\\
\({\;\;\;\;\;\;}\)\({\;\;\;\;\;\;}\)\textbf{/*MWIS for each marginal row*/}\\
\({\;\;\;\;\;\;}\)\({\;\;\;\;\;\;}\)\textbf{FOR} each floor $j \in \{1,2,...,Q\}$\\
\({\;\;\;\;\;\;}\)\({\;\;\;\;\;\;}\)\({\;\;\;\;\;\;}\)View the marginal row in $F_t^j$ as a one-dimensional\\
\({\;\;\;\;\;\;}\)\({\;\;\;\;\;\;}\)\({\;\;\;\;\;\;}\)Matrix graph $\tilde{G}_t^j$ and use Algorithm 2 to find a\\
\({\;\;\;\;\;\;}\)\({\;\;\;\;\;\;}\)\({\;\;\;\;\;\;}\)MWIS $\tilde{S}_t^j$ in it with extra constraints induced by\\ \({\;\;\;\;\;\;}\)\({\;\;\;\;\;\;}\)\({\;\;\;\;\;\;}\)$\bar{S}_t^j$ and $\bar{S}_t^{j-1}$;\\
\({\;\;\;\;\;\;}\)\({\;\;\;\;\;\;}\)\({\;\;\;\;\;\;}\)Set $S_t^j=\bar{S}_t^j \bigcup \tilde{S}_t^j$;\\
\({\;\;\;\;\;\;}\)\({\;\;\;\;\;\;}\)\textbf{END}\\
\\
\({\;\;\;\;\;\;}\)\({\;\;\;\;\;\;}\)\textbf{/*Combine all one-dimensional MWIS*/}\\
\({\;\;\;\;\;\;}\)\({\;\;\;\;\;\;}\)Form a set $S_t=\mathop \bigcup \limits_{j=1}^Q  S_t^j$.\\
\({\;\;\;\;\;\;}\)\textbf{END}\\

\({\;\;\;\;\;\;}\)Choose $S_{c}\in\{S_0,S_1,...,S_{M-1}\}$ that has the maximum \\
\({\;\;\;\;\;\;}\)normalized weighted cardinality.\\
\\
${\;\;\;\;\;\;}$\textbf{/*Assign $c$ to $S_{c}$*/}\\
${\;\;\;\;\;\;}$Use the indicator form \(\mathbf{q}=(q_{m,n}^{i})\) to represent \\
${\;\;\;\;\;\;}$$S_{c}$ and set
\begin{equation}\label{13}
\mathscr{C}(v_{m,n}^i,c)=q_{m,n}^{i},\forall m,n,i
\end{equation}
\textbf{END}\\
\\
Output the solution $\mathscr{C}$.
\end{algorithm}

\subsection{Reuse Ratio Lower Bound and Complexity Analysis}\label{Analysis}
\vspace{0.2cm}
In this subsection, we present the Theorem 2 which analyses the performance of Algorithm 1.\\
\begin{theorem}
Let $\mathscr{C}^*$ be the exact solution for the matrix graph Coloring (MGC) problem in the matrix graph $G$ and let $\bar{f}_m$ be the corresponding maximum weighted reuse ratio. Then Algorithm 1 obtains an approximate solution $S$ with complexity $\mathcal{O}(CK^{L-1}MN)$. Furthermore, the corresponding weighted reuse ratio $\bar{f}$ satisfies
\begin{equation}\label{14}
\bar{f}>\bar{f}_m\cdot \frac{L-1}{L},
\end{equation}
where $C$ is the number of colors. $L$ is the floor height designed beforehand. $K= \max \limits_{m,n} K_{m,n}$ and $K_{m,n}$ denotes the number of independent sets in $G_{m,n}$. \\
\end{theorem}
\begin{proof}
The proof will be divided into three parts. We first show that proving~\eqref{14} can be decomposed into proving the corresponding inequality for each color $c$. Then we prove that the floor division scheme can ensure the inequality for each color $c$. Finally we analyze the computation complexity.\\

To decompose~\eqref{14}, we plug~\eqref{6} into~\eqref{7} and get
\[
\bar f =\frac{\mathop \sum \limits_{v\in V}w_v\cdot\frac{1}{{  C}}\mathop \sum \limits_{c  = 1}^{  C} {\mathscr{C}(v,c )\mu(v,c )}}{\mathop \sum \limits_{v\in V}w_v}.
\]
Then we change the summation order of the numerator and arrive at
\begin{equation}\label{15}
%\begin{split}
\bar f
%&\frac{\frac{1}{{  C}}\mathop \sum \limits_{c  = 1}^{  C}\mathop \sum \limits_{v\in V} w_v {\mathscr{C}(v,c )\mu(v,c )}}{ \mathop \sum \limits_{v\in V} w_v}\\
=\frac{1}{{  C}}\mathop \sum \limits_{c  = 1}^{  C}[\frac{{ \mathop \sum \limits_{v\in V} w_v} \mu(v,c )}{{\mathop \sum \limits_{v\in V} w_v}} \cdot B_{c}],
%\end{split}
\end{equation}
where
\begin{equation}\label{16}
B_{c} =\frac{\mathop \sum \limits_{v\in V} \mathscr{C}(v,c ) w_v {\mu(v,c )}}{ \mathop \sum \limits_{v\in V} w_v \mu(v,c)}.
\end{equation}
For each fixed $c\in \Lambda$, $B_{c}$ is only determined by $\mathscr{C}(v,c),v\in V$, i.e., how this specific color $c$ is assigned to the vertices in $G$. Therefore, optimizing $B_{c}$ has nothing to do with other color assignments. If we use a set \(S_c \subset V\) to denote the vertex set such that $\mathscr{C}(v,c)=1$ and we define weights as~\eqref{12}, then it is easily seen that $B_{c}$ is the normalized weighted cardinality of \(S_c\). Thus, optimizing $\bar f$ in~\eqref{15} can be decomposed into $C$ subproblems and each of them regards maximizing a specific $B_{c}$ by finding a specific MWIS \(S_c\). Then we assign each $c$ to \(S_c\) like~\eqref{13}. As long as we get the approximate MWIS \(S_c\) with a performance guarantee $1-1/L$, we can conclude that~\eqref{14} holds.\\

We next claim that the floor division scheme indeed yields $B_c=|S_c|_N>(1-1/L)|S^*_c|_N$. Define $\mathbf{q}=(q_{m,n}^i)$ as the indicator from of $S_{c}^*$, the MWIS of $G$ with the vertex weights defined as~\eqref{12}. In the following we compare the normalized cardinality of $S_{c}^*$ to the floor-division-based approximate solution $S^{c}$ by induction.\\

As shown in Appendix~\ref{floor_division_scheme}, we have got the floor division scheme $\{F_t^j\}$ beforehand, where $t$ is from $0$ to $L-1$ and $j$ is from $1$ to $Q$. Deleting the marginal row $m(t,j)$ in each floor $F_t^j$, we get a one-dimensional matrix graph $\bar{G}_t^j$ with the row set $\bar{F}_t^j=F_t^j\setminus m(t,j)$, and we can use Algorithm 2 in the Appendix~\ref{Dynamic_Programming} to obtain an exact MWIS solution $\bar{S}_t^j$. We denote this solution in an indicator form $\mathbf{\bar{\theta}}=(\bar{\theta}_{m,n}^i)$. By definition of the MWIS, $\bar{S}_t^j$ must have a larger normalized weighted cardinality than any other independent sets. Recall that $\mathbf{q}$ is the indicator form of $S_{c}^*$, we have, for each $\{F_t^j\}$, that
\[
%\begin{split}
%|\bar{S}_t^j|_N\sum\limits_{m \in \overline F _t^j} {\sum\limits_{n = 1}^N {\sum\limits_{i = 1}^{{l_{m,n}}} {u_{m,n}^i} } }  =
\sum\limits_{m \in \overline F _t^j} {\sum\limits_{n = 1}^N {\sum\limits_{i = 1}^{{l_{m,n}}} {u_{m,n}^i\bar{\theta} _{m,n}^i} } }\ge \sum\limits_{m \in \overline F _t^j} {\sum\limits_{n = 1}^N {\sum\limits_{i = 1}^{{l_{m,n}}} {u_{m,n}^i} } } q_{m,n}^i.
%\end{split}
\]
Summing up the above inequality for all floors $j \in \{1,...,Q\}$ in the $t$-th floor division, we obtain
\[\sum\limits_{j = 1}^Q {\sum\limits_{m \in \overline F _t^j} {\sum\limits_{n = 1}^N {\sum\limits_{i = 1}^{{l_{m,n}}} {u_{m,n}^i} } } } \bar{\theta} _{m,n}^i \ge \sum\limits_{j = 1}^Q {\sum\limits_{m \in \overline F _t^j} {\sum\limits_{n = 1}^N {\sum\limits_{i = 1}^{{l_{m,n}}} {u_{m,n}^i} } } q_{m,n}^i}. \]
Defining $\bar{S}_t=\mathop \bigcup \limits_{j=1}^Q  \bar{S}_t^j$, we have
\[
\begin{split}
|{\overline S _t}|_N \ge &\frac{{\sum\limits_{j = 1}^Q {\sum\limits_{m \in \overline F _t^j} {\sum\limits_{n = 1}^N {\sum\limits_{i = 1}^{{l_{m,n}}} {u_{m,n}^i} } } q_{m,n}^i} }}{{\sum\limits_{j = 1}^Q {\sum\limits_{m \in \overline F _t^j} {\sum\limits_{n = 1}^N {\sum\limits_{i = 1}^{{l_{m,n}}} {u_{m,n}^i} } } } }}\\
> &\frac{1}{\Sigma }\sum\limits_{j = 1}^Q {\sum\limits_{m \in \overline F _t^j} {\sum\limits_{n = 1}^N {\sum\limits_{i = 1}^{{l_{m,n}}} {u_{m,n}^i} } } q_{m,n}^i},
\end{split}
\]
where
\[
\Sigma  = \sum\limits_{m = 1}^M {\sum\limits_{n = 1}^N {\sum\limits_{i = 1}^{{l_{m,n}}} {u_{m,n}^i} } }.
\]
After adding new nodes in $\bar{S}_t$, we get an independent set $S_t$ with larger normalized cardinality, thus we have
\[|S_t|_N \cdot \Sigma  > |{\overline S _t}|_N\cdot \Sigma >\sum\limits_{j = 1}^Q {\sum\limits_{m \in \overline F _t^j} {\sum\limits_{n = 1}^N {\sum\limits_{i = 1}^{{l_{m,n}}} {u_{m,n}^i} } } q_{m,n}^i}. \]
Lemma~\ref{fl_div_lmm} ensures that each row $j$ appears in exactly $L-1$ different floor divisions (except being the marginal row only once), so if we sum the above equation for all $t$, we arrive at
\[
\begin{split}
\sum\limits_{t = 0}^{L - 1} {|S_t|_N \cdot \Sigma }  \ge& \sum\limits_{t = 0}^{L - 1} {\sum\limits_{j = 1}^Q {\sum\limits_{m \in \overline F _t^j} {\sum\limits_{n = 1}^N {\sum\limits_{i = 1}^{{l_{m,n}}} {u_{m,n}^i} } } q_{m,n}^i} }\\
=& (L - 1)\sum\limits_{m = 1}^M {\sum\limits_{n = 1}^N {\sum\limits_{i = 1}^{{l_{m,n}}} {u_{m,n}^iq_{m,n}^i} } }\\
= &(L - 1)|{S^*}|_N \cdot \Sigma.
\end{split}
\]
Dividing both sides with $\Sigma L$ yields
\begin{equation}\label{17}
\frac{1}{L}\sum\limits_{t = 0}^{L - 1} {|S_t|}_N  > \frac{{L - 1}}{L}|{S^*}|_N.
\end{equation}
If we choose $t=t^*$ s.t. $S_{t^*}$ has the largest normalized cardinality, we will have
\begin{equation}\label{18}
|S_{{t^*}}|_N > \frac{{L - 1}}{L}|{S_c^*}|_N.
\end{equation}
$S_{{t^*}}$ is exactly our approximate solution for $S_{c}$. Thus, we know that $B_{c}$ is guaranteed to obtain the $1-1/L$ of the optimal value. And based on~\eqref{15}, we know that~\eqref{14} holds.\\

The complexity scales like the following: For each $c$, we need to find the MWIS $S_c$, which is further decomposed into totally $QL$ subproblems. Each problem is solving the MWIS problem in a one-dimensional matrix graph. Based on Lemma~\ref{lmm_DP}, we can show that each problem will be completely solved with complexity $\mathcal{O}(K^{L-1}N)$. Therefore, the final problem will be solved in $\mathcal{O}(C QLK^{L-1}N)= \mathcal{O}(C K^{L-1}MN)$.\\

The complexity $\mathcal{O}(K^{L-1}N)$ is obtained like the following. In fact, each cell contains at most $K$ Independent Sets. Based on the IS decomposition~\eqref{9}, we know that if we view each $M_j$-by-$N$ subgraph as a one-dimensional matrix graph, then one big cell is constituted of $M_j$ cells vertically, and each big cell contains at most $K^{M_j}$ Independent Sets. We know from Lemma~\ref{fl_div_lmm} that $M_j<L-1$, thus, each sub-problem can be solved with complexity $\mathcal{O}(K^{L-1}N)$. \\ \end{proof}
\begin{remark}Theorem 2 characterizes the tradeoff between computation complexity and efficiency that we can get, which forms a theoretical foundation to get the performance-guaranteed coloring scheme in a matrix graph. We have made the statement that matrix graphs are especially computing-efficient for small cell graphs. Now it is supported here. Since $K$ is a very small number, $\mathcal{O}(K^L)$ will not be especially large if the floor height $L$ is not that large. Moreover, if inter-cell interferences are high, the complexity $\mathcal{O}(K^L)$ further shrinks due to the branch trimming in finding one-dimensional MWISs (The Dynamic Programming in Lemma~\ref{lmm_DP}). In practice, if we choose $L=5$, then based on Theorem 2, we can get a performance guaranteed to be better than $1-1/5=80\%$ of the optimal one. Moreover, simulation results suggest that this lower bound is quite loose. Usually the performance reaches more than $95\%$. A tighter bound is our goal in the future.\\
\end{remark}

\begin{remark}
One might be concerned with the computational complexity which grows exponentially with the parameter $L$ to achieve the frequency allocation bound. However, this $(\mathcal{O}(\frac{L-1}{L}),\mathcal{O}(K^L))$ performance-complexity tradeoff is inevitable due to the NP-Completeness. In fact, if we get a $(\mathcal{O}(\frac{L-1}{L}),\mathcal{O}(L^\alpha))$ tradeoff in the MGC problem (which is defined in the complexity class FPTAS~\cite{21}) and $L$ could go to infinity, we can simply set $L$ to be the same as the number of vertices in the graph, set $C=1$ and set all weights to be 1, which finally yields an approximate maximum independent set solution that hits $1/L$ to the bound with polynomial complexity of the network size. However, since $L$ is the number of vertices, the smallest granularity of a Maximum Independent Set problem (specific MWIS problem when all weights are 1) is now $1/L$. Thus, the approximate solution is exactly the same as the optimal one. This contradicts with the general belief that in NP-complete problems, we cannot find any polynomial-time solution that achieves the bound. Nonetheless, one can still explore new ways to lower the base $K$ of $\mathcal{O}(K^L)$ in order to get the best exponential.
\end{remark}

\section{Simulations}\label{Simulations}
\vspace{0.2cm}
In this section, simulation results are obtained for large-scale small-cell networks. The test bed is set to be a randomly generated $M$-by-$N$ matrix graph as follows. We first generate a geometric random graph with a Poisson point process with density $\lambda$ and each two vertices are connected if their distance is within $2r$. Then we separate the graph into totally 12000 cells, with $M=60$ and $N=200$. In the first simulation we will change $N$ from 1 to 200 to view the convergence result. After that we set $N=200$ to view the performance variation with other parameters. No matter $N$ changes or not, $M$ and $N$ are set before generating the matrix graph. The expectation number of vertices in each cell will be $V_d=4\lambda r^2$, indicating the \emph{vertex density}. In order to view the performance under different interference intensity, each edge is erased with probability $1-E_d\in[0,1]$. The parameter $E_d$ is called the \emph{edge density}. Assume we have $C=6$ colors, which is the same setting in~\cite{16}. The color number does not affect the conclusion. In order to compare with other algorithms~\cite{15}\cite{16}, we simply set color weight $\mu(v,c)\in\{0,1\}$, which equals to 1 with probability $p_f$. Thus, the equivalent vertex density is actually  $V_d\cdot p_f$, because we never assign a color to communication links with 0 weights. In the following when we refer to vertex density, we actually refer to $V_d\cdot p_f$. After constructing the matrix graph, the Algorithm 1 is simulated and the frequency allocation scheme is obtained. The performance criterion is the weighted reuse ratio defined in~\eqref{7}. We assume all vertex weights are 1, which does not affect the simulation results. So, this criterion now just equals to the average ratio of frequency bands that is used by each communication link, which directly shows the resource reuse efficiency.\\
\begin{figure}
\centering
\includegraphics[scale=0.65]{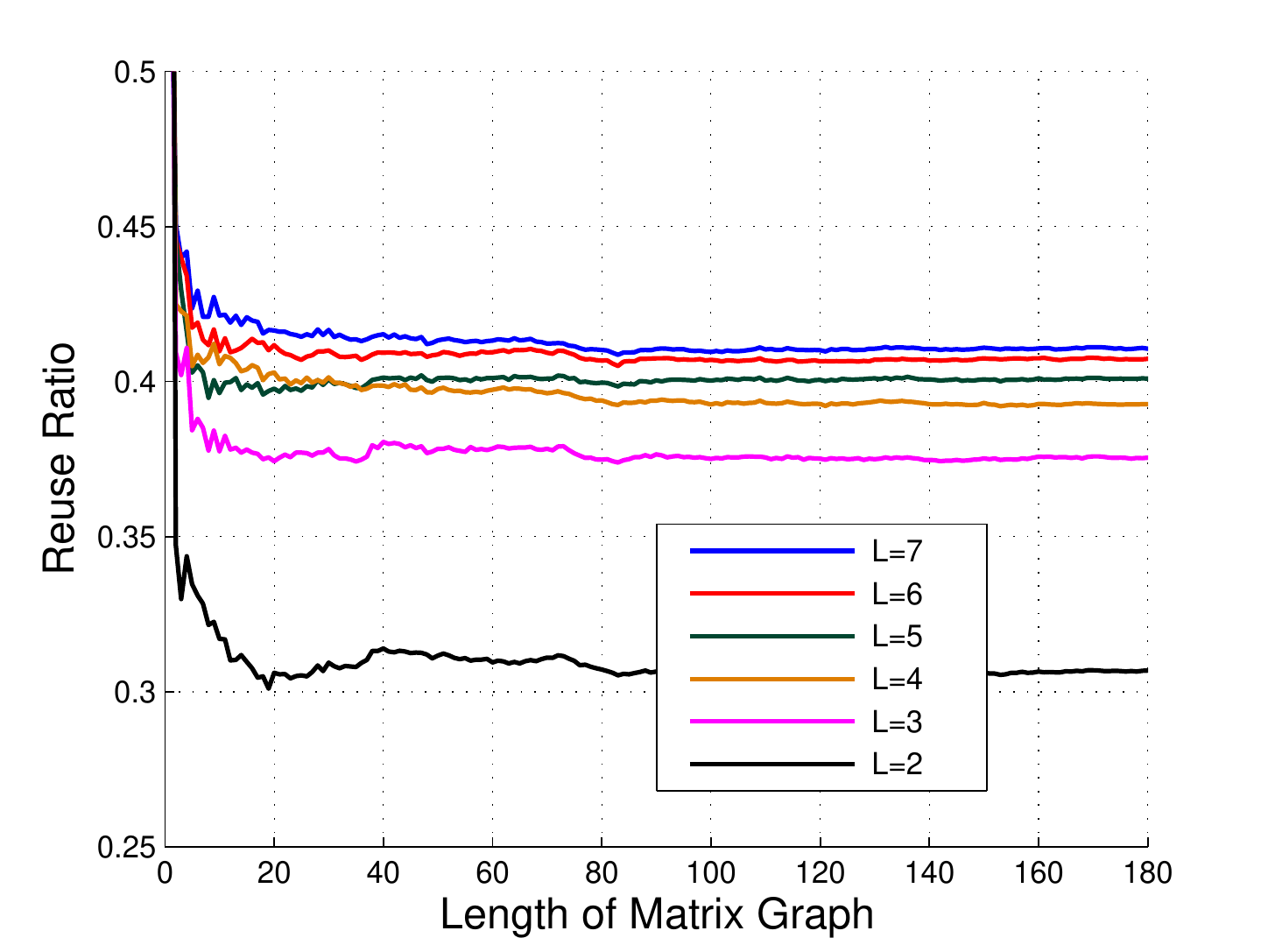}
\caption{{Convergence of the weighted reuse ratio when vertex density=1.6, edge density=0.6.}}\label{fig1，2}
\end{figure}
\begin{figure}
\centering
\includegraphics[scale=0.65]{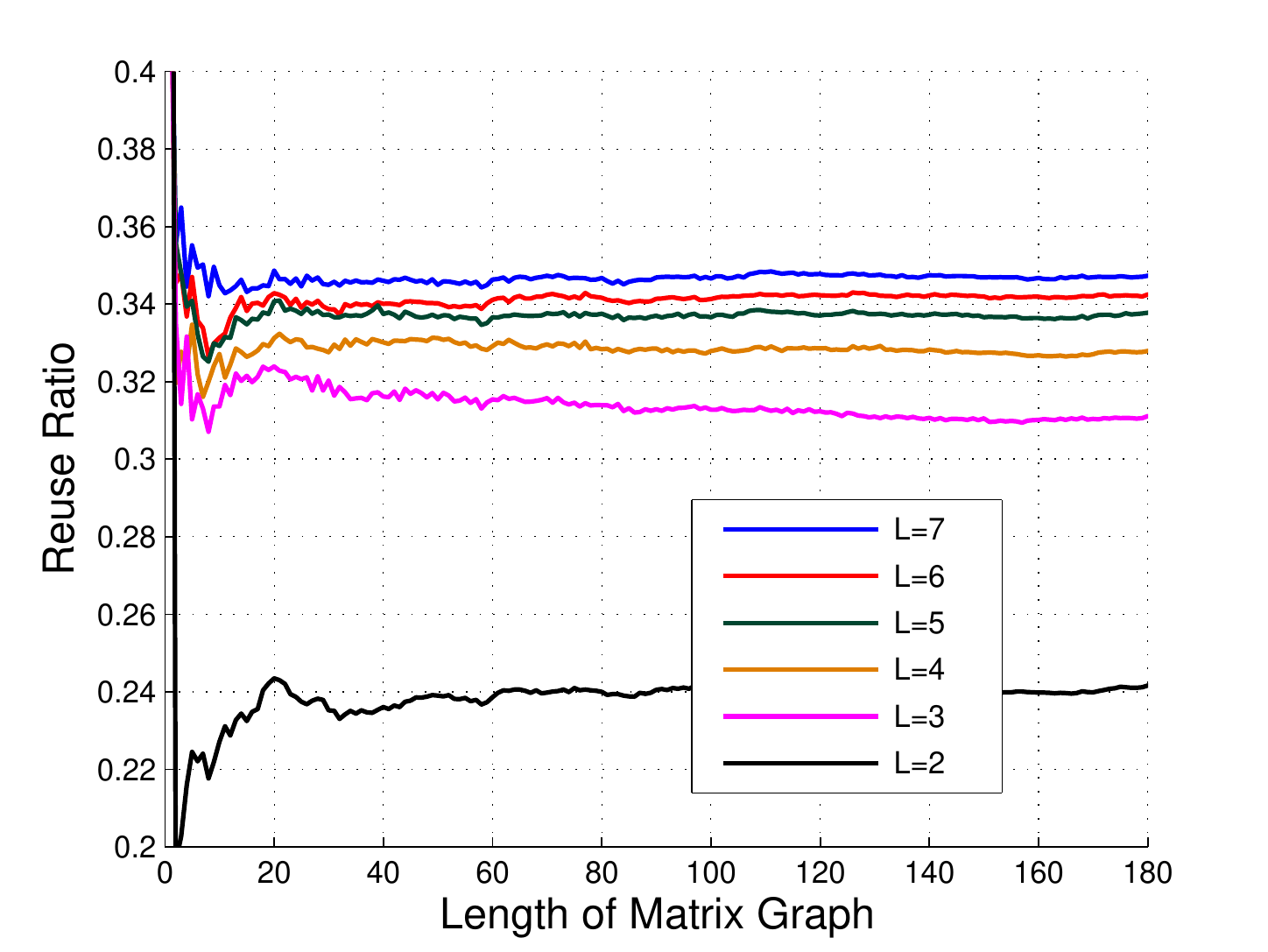}\vspace{-0.3cm}
\caption{{Convergence of the weighted reuse ratio when vertex density=1.6, edge density=0.8.}}\vspace{-0.5cm}\label{fig1，2}
\end{figure}

In Fig. 4 and 5, the horizontal axis is the length $N$ of the matrix graph. We set $M=60$ and changes $N$ from 1 to 200, while taking down the weighted reuse ratio obtained by Algorithm 1. In these two figures, the vertex density is set to be 1.6 and the edge density is 0.6 and 0.8 respectively. We find that when $N$ goes large, each curve converges to a constant value. For different curves (with different floor height $L$), all curves uniformly converge (simultaneously for each $N$) to a limit. This limit is the theoretical limit of frequency allocation.\\

In order to support Theorem 2 which says that the solution obtained by Algorithm 1 has at most a $1/L$ gap to the optimal solution, we illustrates the performance when the floor height $L$ goes large, under different vertex and edge densities (as shown in Fig. 3). It is clear that when L increases, each curve converges to a limit. Thus, we can approximately tell the theoretical limit of frequency allocation, despite the fact that telling the exact value has been proved to be NP-complete.\\

A more interesting result is that, when vertex density and edge density increases, this limit shrinks. This is intuitively right because as interference relationships become complicated, the available resources to be reused decreases. We conjecture that this limit, on a randomly generated large scale network, only depends on vertex density and edge density. A meaningful future work is to investigate this conjecture, which can ultimately tell the frequency reuse limit.\\
\begin{figure}
\centering
\includegraphics[scale=0.65]{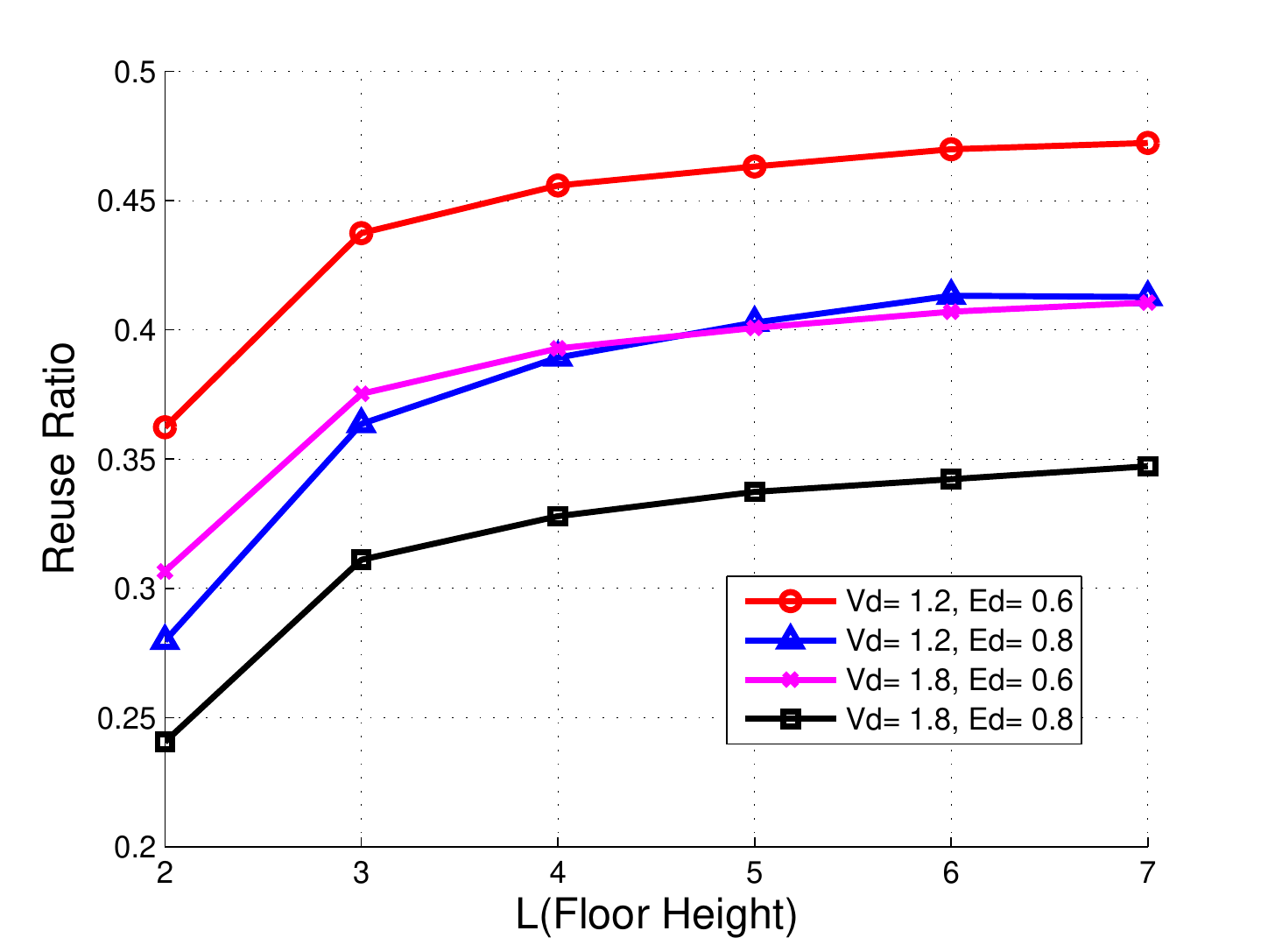}\vspace{-0.3cm}
\caption{{The reuse ratio converges to the optimal value when $L$ increases.}\vspace{-0.5cm}}\label{fig1，2}
\end{figure}

In Fig. 7 and 8 we show the performance comparison of the Algorithm 1 with three other algorithms. GB-DFR is a graph based heuristic proposed in~\cite{16}, which generalized the conception of saturation-degree graph coloring in~\cite{13} and got good performance in cellular system simulations. GLC is the Greedy List-Coloring proposed in~\cite{15}. It is simple and efficient. We find that our algorithm performs gradually better when edge density and vertex density increases. This is common since graph-based algorithms usually have good performance in degree-bounded graphs. But when interference become complicated, there is no guarantee that they perform well. By the way, after one color is assigned to a vertex, both GB-DFR and GLC have sorting in the whole network, which drives the complexity to $\mathcal{O}(MN\bar{f}C\cdot MN\log MN)$, where $\bar f$ is the weighted reuse ratio and $M$-by-$N$ is the network size. When network goes large, this becomes impractical. SFR is called Soft Frequency Reuse~\cite{16}, which uses different reuse factors in cell edge and cell center. In our matrix graph, we just consider the cell center to be vertices that do not interfere with the neighboring cells. Since SFR is essentially a grid-model algorithm, it does not perform quite well in our tests. However, when interference is quite large (edge density reaches 0.8), it has excellent performance. We suspect that this is because when edge density reaches some threshold, interference management schemes does not have much gain compared to interference avoidance schemes.\\
\begin{figure}
\centering
\includegraphics[scale=0.65]{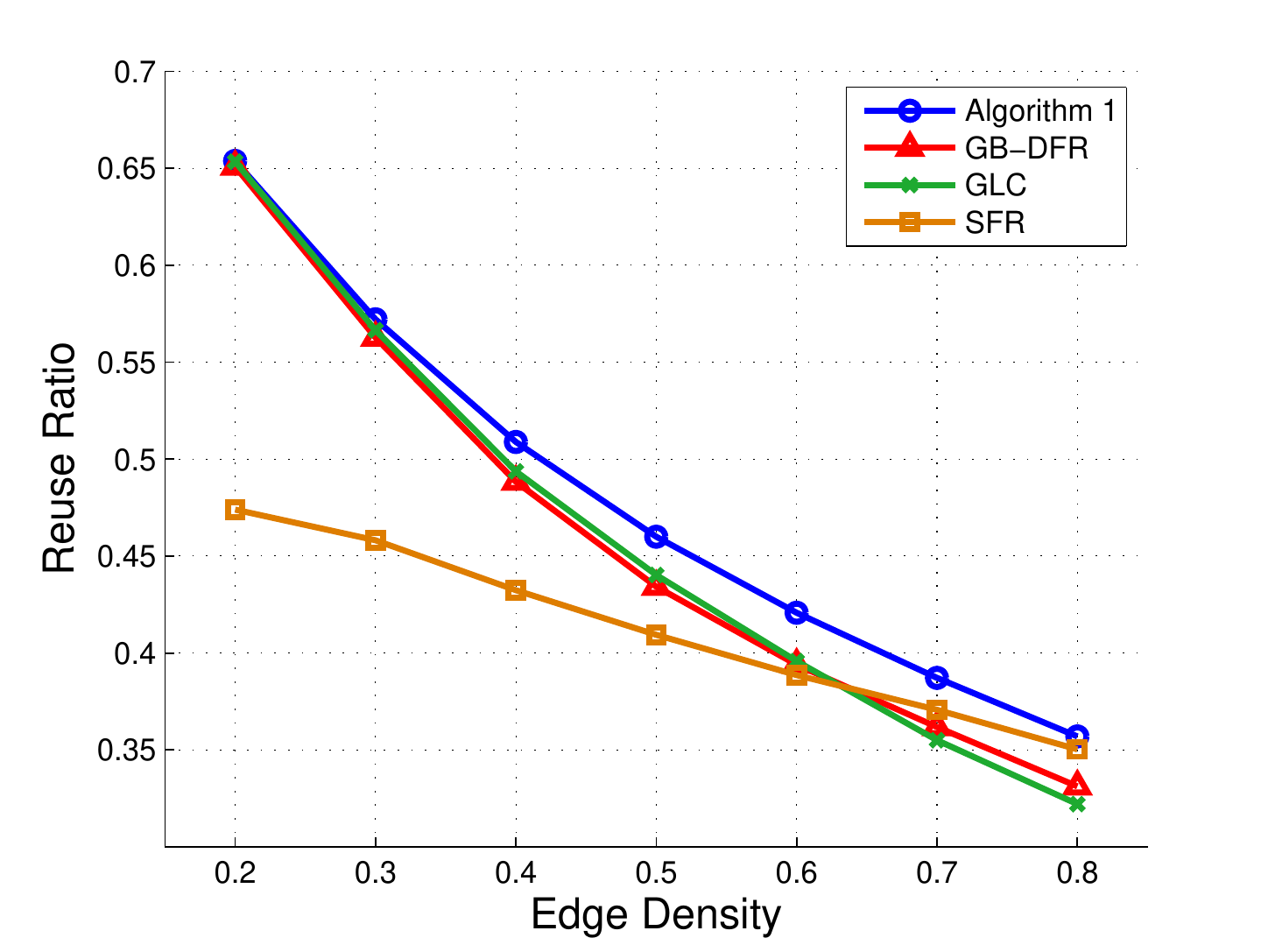}
\caption{{Performance Comparison of Algorithm 1 with other algorithms when vertex density=1.6}}\label{fig1，2}
\end{figure}
\begin{figure}
\centering
\includegraphics[scale=0.65]{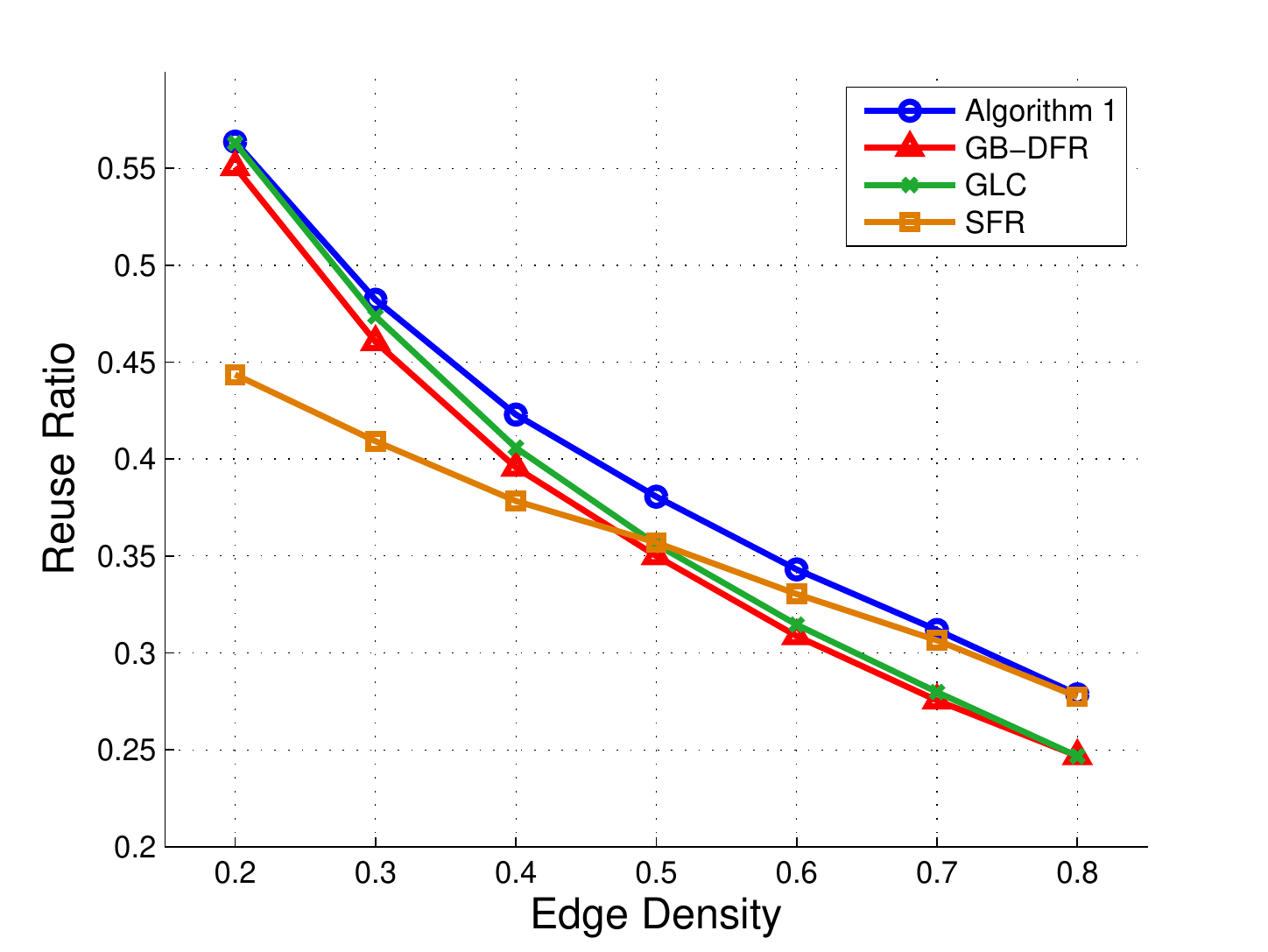}
\caption{{Performance Comparison of Algorithm 1 with other algorithms when vertex density=2.4}}\label{fig1，2}
\end{figure}
\vspace{-0.8cm}
\section{Conclusions}
In this paper we are focusing on the ultimate limit of frequency allocation in a 5G network. To study this problem, we proposed a matrix graph model and constructed an analytical framework combining matrix graph coloring (MGC) and maximum weighted independent set (MWIS), based on properties of large-scale small-cell networks. Utilizing this model, we obtain an approximation algorithm that achieves a bounded gap to the optimal performance with a complexity growing linearly with the network size, despite the NP-completeness of the MGC problem. Therefore, if we could build a proper matrix graph, we can find the nearly-optimal way to allocate resources like frequencies and time slots. This is in contrast with conventional graph-coloring based heuristics which usually have no guarantee on performance. Moreover, the proposed scheduling algorithm has lower computational complexity if cells are smaller and inter-cell interference are more complicated. Thus, we conclude that frequency allocation in high-interference small-cell networks can be carried out efficiently and the small-cell networks are indeed practical for the future 5G network construction. Although rich simulation results support our theories, we are still interested in further improving them. Since our simulations are carried out on random graphs, according to our observations, a random-graph analytical way to derive performance bound might exist. If this is the case, we could directly calculate the performance bound expectation regardless of the NP-completeness of finding a concrete coloring scheme, even without carrying out the approximation algorithms. At least, the bound of $(L-1)/L$ in Theorem 2 could be further tightened due to the law of large numbers in a random graph.\\
\appendices
\section{Proof of Lemma~\ref{SG_lmm}}\label{Graph_Partition_performance_loss}
Suppose we have excluded $E_c$ edges in the matrix graph construction. Then, the validity check will drop at most $E_c$ users for each frequency band. If we denote the optimal reuse ratio after matrix graph construction as $\bar{f}_m^{*}$, then it holds that
\[\bar{f}_m>\bar{f}_m^{*}-\frac{E_c}{|V|}.\]
It is clear that $\bar{f}_m^{*}>\bar{f}^{*}$ because dropping edges (lowering interference) can only increase reuse ratio, so we have
\[\bar{f}_m>\bar{f}^{*}-\frac{E_c}{|V|}.\]
Thus, in order to prove~\eqref{Drop_Edge_optimal}, it suffices to show that
\begin{equation}\label{Drop_Edge}
  E[\frac{E_c}{|V|}]<\frac{1}{2}\lambda \int_{\mathbf{x}\in\Omega}g(\mathbf{x})d\mathbf{x}.
\end{equation}
Define the number of excluded edges connecting node $v_i$ by $E_i$, and denote $|V|$ by $N$, then conditioning on $N=n$, it follows that
\begin{equation}\label{Derive1}
\begin{split}
E[\frac{E_c}{|V|}]=&E\{ E[\frac{E_c}{N}|N=n] \}\\
%\overset{(a)}{=}&E\{E[\frac{1}{N}\cdot\frac{1}{2}\sum_{i=1}^{n}E_i|N=n]\}\\
\overset{(a)}{=}&E\{\frac{1}{2N}\sum_{i=1}^{n}E[E_i|N=n]\},
\end{split}
\end{equation}
where (a) is true because each excluded edge is counted twice in the enumeration. For a specific vertex $v_i$, it holds that
\[E_i=\sum_{j=1,j\neq i}^n X_{ij},\]
where
\[
X_{ij} = \left\{ {\begin{array}{*{20}{c}}
{1,}\\
{0,}
\end{array}\begin{array}{*{20}{c}}
\text{($v_i,v_j)\in E$ but located in non-adjacent cells,}\\
\text{otherwise.}
\end{array}} \right.
\]
It follows that
\[E[E_i|N=n]=\sum_{j=1,j\neq i}^n \Pr[X_{ij}=1|N=n].\]
For a Poisson point process with number of vertices fixed to be $n$, the $n$ vertices will follow the independent identical uniform distribution over the whole rectangular area. Thus,~\eqref{Derive1} can be further simplified to
\begin{equation}\label{Derive2}
\begin{split}
E[\frac{E_c}{|V|}]=&E\{\frac{1}{2N}\sum_{i=1}^{n}\sum_{j=1,j\neq i}^n \Pr[X_{ij}=1|N=n]\}\\
\overset{(a)}{=}&E\{\frac{N-1}{2}\Pr[X_{12}=1|N=n]\}.
\end{split}
\end{equation}
The equality (a) follows from that $X_{i,j},\forall i\neq j$ are identically distributed. Therefore, we can focus on two specific nodes $v_1$ and $v_2$ and look for an upper bound for $\Pr[X_{12}=1|N=n]$. Since $X_{12}=1$ only when $v_i$ and $v_j$ are in non-adjacent cells, $X_{12}=0$ surely if horizontal distance or vertical distance between $v_1$ and $v_2$ are both smaller than $a$. Suppose $v_1$ is located at $\mathbf{x}_1=(x_1^1,x_1^2)$ and $v_2$ is located at $\mathbf{x}_2=(x_2^1,x_2^2)$ and $\mathbf{x}_1,\mathbf{x}_2\in \Omega_0$, which is the rectangular area shown in Fig.~\ref{m_graph}a. Then the probability can be written as
\[\Pr[X_{12}=1|N=n]<\int_{\Omega_1}\frac{1}{|\Omega_0|^2}g(\mathbf{x_1-x_2})d\mathbf{x_1}d\mathbf{x_2},\]
where $\Omega_1=\{(\mathbf{x_1},\mathbf{x_2})\in \Omega_0^2||x_1^1-x_2^1|>a$ or $|x_1^2-x_2^2|>a\}$ and $|\Omega_0|$ is the area of $\Omega_0$. Replacing variable $\mathbf{x}_2$ with $\mathbf{y}_2=\mathbf{x}_2-\mathbf{x}_1$, we get
\[\Pr[X_{12}=1|N=n]<\int_{\Omega_2}\frac{1}{|\Omega_0|^2}g(\mathbf{y_2})d\mathbf{x_1}d\mathbf{y_2},\]
where
\[\begin{split}
\Omega_2=&\{(\mathbf{x_1},\mathbf{y_2})\in \Omega_0\times\mathbb{R}^2|\mathbf{y_2}\in \Omega_0+\mathbf{x_1},|y_2^1|>a\text{ or }|y_2^2|>a\}\\
\subset&\{(\mathbf{x_1},\mathbf{y_2})\in \Omega_0\times\mathbb{R}^2||y_2^1|>a\text{ or }|y_2^2|>a\}=\Omega_0\times\Omega.
\end{split}\]
Thus, we have
\[\begin{split}
\Pr[X_{12}=1|N=n]<&\int_{\Omega_0\times\Omega}\frac{1}{|\Omega_0|^2}g(\mathbf{y_2})d\mathbf{x_1}d\mathbf{y_2}\\
=&\int_{\Omega}\frac{1}{|\Omega_0|}g(\mathbf{y_2})d\mathbf{y_2}.
\end{split}\]
Plugging this inequality into~\eqref{Derive2}, we get
\[\begin{split}
E[\frac{E_c}{|V|}]<&E[\frac{N-1}{2}]\int_{\Omega}\frac{1}{|\Omega_0|}g(\mathbf{x})d\mathbf{x}\\
\overset{(a)}{=}&\frac{\lambda|\Omega_0|-1}{2}\int_{\Omega}\frac{1}{|\Omega_0|}g(\mathbf{x})d\mathbf{x}<
\frac{\lambda}{2}\int_{\Omega}g(\mathbf{x})d\mathbf{x},
\end{split}\]
where step (a) holds because $N$ is a Poisson process with mean $\lambda|\Omega_0|$. This concludes the proof.
\section{Proof of Theorem 1: MGC in a Matrix Graph is NP-Complete}\label{Proving_NP_complete}
\begin{figure}
  \centering
  \includegraphics[scale=0.3]{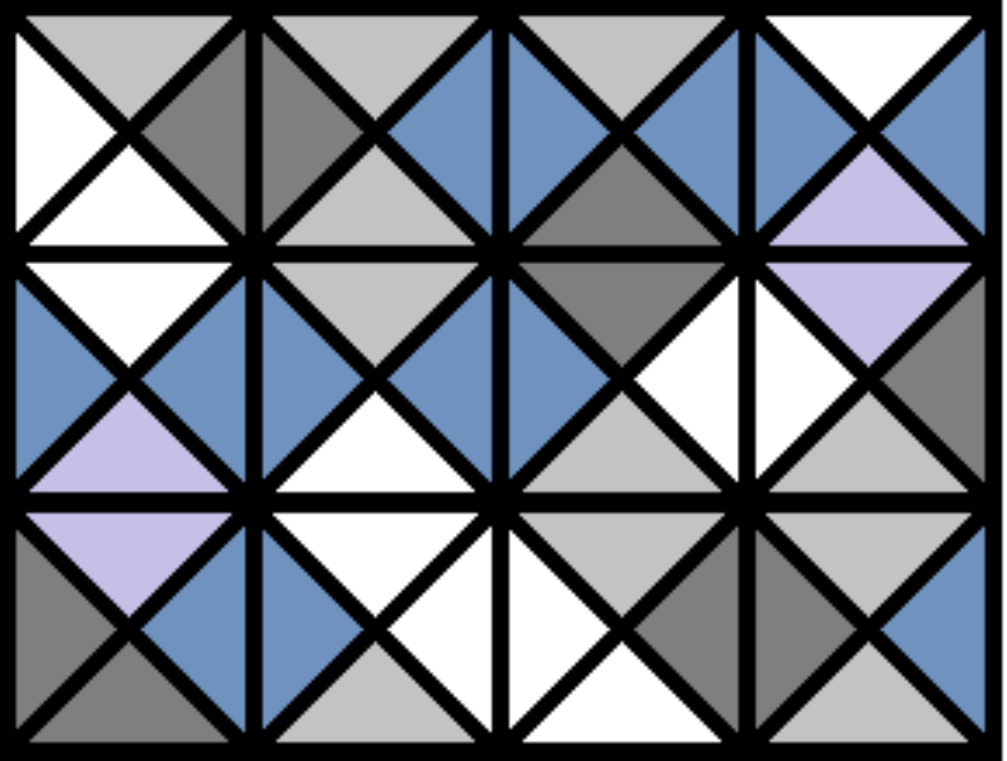}\\
  \caption{{A square tilling of a 3-by-4 finite square with Wang tiles. Any two neighboring tiles have the same color on the common edge.}}\label{fig:graph}\vspace{-0.2cm}
\end{figure}
Firstly, another NP-complete problem, the square tilling~\cite{9}, can be reduced to the MGC problem. Secondly, the MGC problem can be reduced to the maximum weighted independent set problem~\cite{8}. The second part will be shown in the analysis part and we will prove the first part.\\

Since the MWIS problem can be viewed as the MGC problem with one color, the MWIS on a matrix graph can be reduced to the MGC problem. Meanwhile, the MGC problem can be reduced to the general MWIS problem trivially by assigning each color to a MWIS. Thus, if we have proved the first statement, then MGC is NP-complete. The Wang tilling problem~\cite{23}\cite{24} is a classic unsolvable combinatorial problem. A Wang tile is a square with its four edges, namely north-,east-,west- and south-edges colored by a set of colors. Now assume that we have a set of Wang tiles \(W = \{w_1,w_2,...,w_{l-1},w_l\}\). A tiling \(T\) is said to be valid, if neighboring tiles has the same color. The following Figure shows an example of Wang tiling of a 3-by-4 square.In~\cite{23}, the author stated that whether a given set of Wang tiles can validly tile a \(M \times N\) square is NP-complete with the size of square. The author has not given the proof in~\cite{23}, but a following paper~\cite{24} proved a special case of original problem to be NP-complete. So the NP-completeness of the original tilling problem in~\cite{24} is also ensured.

Assume we have a square lattice denoted by \(\{(m,n)\}_{m=1,n=1}^{M,N}\) to be tiled by the given tile set $W$. To reduce the tilling of a square to a MWIS problem in a matrix graph $G=(V,E)$, we first construct the corresponding graph. Writing
\[
{ {V}} = \mathop \bigcup \limits_{m,n =  1}^{M,N}  {{ {V}}_{m,n}}
\]
for the vertex set of $G$, where each \({{ {V}}_{m,n}}=\{v_{m,n}^i\}_{i=1}^{l}\) is the vertex set of a L-complete graph \(G_{m,n}=K_l\). Each vertex \(v_{m,n}^i\) is associated with a tile \(w_i\). For two horizontally neighboring vertex sets, for example, \(V_{m,n}\) and \(V_{m+1,n}\), \((v_{m,n}^i,v_{m+1,n}^j)\in E\) if and only if tile \(w_i\)'s east edge does not match tile \(w_j\)'s west edge when they are respectively put at lattice point \((m,n)\) and \((m+1,n)\). Similarly, for two vertically neighboring vertex sets \(V_{m,n}\) and \(V_{m,n+1}\), \((v_{m,n}^i,v_{m,n+1}^j)\in E\) if and only if tile \(w_i\)'s south edge does not match tile \(w_j\)'s north edge when they are respectively put at lattice points \((m,n)\) and \((m,n+1)\). One can easily check that $G=(V,E)$ constructed above is a matrix graph consistent with definition 1.\\

Next we show that tiling the $M$-by-$N$ square can be reduced to finding a Maximum Weighted Independent Set in $G$ with vertex weights $u_{m,n}^i=1,\forall m,n,i$. Since each cell of the matrix graph is a complete graph \(K_l\), we can only pick up one vertex from each cell. If the maximum weighted independent set that we find in $G$ coincidentally picks up one vertex, with the index \(i(m,n)\), in each cell \({{ {G}}_{m,n}}\), then we can construct a tilling \(T(m,n)=i(m,n)\) of the square. Since the conflictions between two edges in the matrix graph indicates the mismatching between corresponding tiles, we know that this tilling \(T(m,n)=i(m,n)\) has no mismatching and is valid. As a result, the square tilling problem can be reduced to tell if the maximum weighted independent set in this matrix graph $G$ has a normalized weighted cardinality \(1/l\)(1 vertex from $l$ vertices in each cell). Since the tilling problem is NP-complete, the general MWIS problem in a matrix graph has the same difficulty.

\section{Proof of Lemma~\ref{lmm_DP}: MWIS in a one-dimensional Matrix Graph can be solved in Linear Time}\label{Dynamic_Programming}
\vspace{0.2cm}
In this section we show that MWIS problem in a one-dimensional matrix graph can be solved completely in linear time. Before giving out the dynamic programming algorithm, we need to review some properties of a one-dimensional matrix graph. We call a one-dimensional matrix graph is a \emph{Vector Graph}. Solving MWIS in a Vector Graph can give us convenience on solving MWIS in general matrix graphs. Moreover, apart from this convenience, we have mentioned that one-dimensional cellular network itself is of particular practical interests. Similar to Definition 1, we have\\
\begin{definition}
A Graph \(G=(V,E)\) is a Vector Graph if\\
\begin{equation}\label{19}
{ {V}} = \mathop \bigcup \limits_{n =  1 }^{N}  {V_n}
\end{equation}
\begin{equation}\label{20}
V_{n}=\{v_{n}^i\}_{i=1}^{l_{n}}
\end{equation}
An edge \((v_{n_1}^i,v_{n_2}^j)\in E\) only if
\begin{equation}\label{21}
|n_1-n_2|\le 1
\end{equation}

\end{definition}
\begin{algorithm}\caption{Finding MWIS in a one-dimensional matrix graph with constraints $Y$}\label{alg:B}
\textbf{Input}: A Vector Graph $G=(V,E)$, vertex weights $\mathbf{u}=(u_{n}^i)$, constraints $Y=\{Y_n\}_{n=1}^N,Y_n\subset X_n,\forall n$\\
\textbf{Output}: A MWIS $S^*$ which optimizes~\eqref{17}.\\
\\
Initialize\\
For all \(k_1\) s.t. \(\alpha_1^{k_1}\in X_1\)\\
\({\;\;\;\;\;\;}\)if   \(\alpha_1^{k_1}\notin Y_1\), set \(\wp^{k_1}_{(1)}=\emptyset\);\\
\({\;\;\;\;\;\;}\)else set \(\wp^{k_1}_{(1)}=(S_1)=(\alpha_1^{k_1})\);\\
end\\
\\
For \(n\) from \(2\) to \(N\)\\
\({\;\;\;\;\;\;}\)For all \(k_n\) s.t. \(\alpha_n^{k_n}\in X_n\)\\
\({\;\;\;\;\;\;}\)\({\;\;\;\;\;\;}\)if \(\alpha_n^{k_n}\notin Y_n\) set \(\wp^{k_n}_{(n)}=\emptyset\){\;\;\;\;\;\;}(\textbf{Extra Constraints})\\
\({\;\;\;\;\;\;}\)\({\;\;\;\;\;\;}\)else find \(l^{*}\in \{1,...,K_{n-1}\}\) s.t. \\
\({\;\;\;\;\;\;}\)\({\;\;\;\;\;\;\;\;}\)1).\((\alpha_{n-1}^{l^{*}},\alpha_n^{k_n})\in R_{n-1,n}\){\;\;\;\;}(\textbf{1D Constraints})\\
\({\;\;\;\;\;\;}\)\({\;\;\;\;\;\;\;\;}\)2).$\wp^{l^{*}}_{(n-1)}\neq \emptyset$\\
\({\;\;\;\;\;\;}\)\({\;\;\;\;\;\;\;\;}\)3).$l^{*}$ maximizes $ |\wp^{l^{*}}_{(n-1)}|_N $  (\textbf{Bellman Equation})\\
\({\;\;\;\;\;\;}\)\({\;\;\;\;\;\;}\)set \(\wp^{k_n}_{(n)}=(\wp^{l^{*}}_{(n-1)}S_n)=(\wp^{l^{*}}_{(n-1)}\alpha_n^{k_n})\).\\
\({\;\;\;\;\;\;}\)end\\
end\\

Find $k^{*}\in \{1,...,K_N\}$ that maximizes \(|\wp^{k}_{(N)}|_N\). \(\wp^{k^{*}}_{(N)}\) is the maximum weighted independent set that we are seeking for.\\
Output $S^{*}=\wp^{k^{*}}_{(N)}$.
\end{algorithm}
We use the notation $G_{n}$ to denote the cell that contains $V_{n}$. As a counterpart to~\eqref{8}, we decompose an independent set $S$ in $G$ by
\begin{equation}\label{22}
S = \mathop \bigcup \limits_{n=1}^N  {{ {S}}_{n}}
\end{equation}
and the Maximum Weighted Independent Set problem is aimed at maximizing
\begin{equation}\label{23}
\left| { {S}} \right|_N=\frac{\mathop \sum \limits_{n = 1}^{N} \sum \limits_{i=1}^{l_{n}} q_{n}^i u_{n}^i}{\mathop \sum \limits_{n = 1}^{N} \sum \limits_{i=1}^{l_{n}} u_{n}^i}
\end{equation}
where $\mathbf{u}=(u_{n}^i)$ are the vertex weights.\\

Then we define the sequence representation of an independent set. Noticing that if $S$ is an independent set of $G$, then for \(\forall n\), \(S_n\) is an independent set of the corresponding cell \(G_n\). We denote all possible independent sets of \(G_n\) by \(X_n=\{\alpha_n^1,\alpha_n^2,...,\alpha_n^{K_n}\}\). Suppose that \(S_n=\alpha_n^{k_n}\) for \(\forall n\), $S$ can be written in a \emph{N-sequence representation}
\begin{equation}\label{24}
\begin{split}
S=&(S_1S_2...S_n...S_N)\\
=&(\alpha_1^{k_1}\alpha_2^{k_2}...\alpha_n^{k_n}...\alpha_N^{k_N}),k_n\in \{1,2,...,K_n\},\forall n
\end{split}
\end{equation}
For simplicity of notation, we use the same letter $S$ for this sequence. When mentioning the normalized weighted cardinality (NWC) of a sequence $S$, we refer to the NWC of the corresponding independent set.\\\\
For each two adjacent cells \(G_n\) and \(G_{n+1}\), we define \(G_{n,n+1}\) as the induced graph containing \(G_n\) and \(G_{n+1}\), i.e. the graph that contains \(G_n\), \(G_{n+1}\) and the confliction edges between them. Then we define a relation
\begin{equation}\label{25}
\begin{split}
R_{n,n+1}=\{(\alpha,\beta)|&\alpha\in X_n,\beta\in X_{n+1}, 2-sequence{\;}(S_{n}S_{n+1})\\ &=(\alpha\beta){\;}is{\;}an{\;}independent{\;}set{\;}of{\;}G_{n,n+1}\}
\end{split}
\end{equation}
where $X_n$ still denotes all possible independent sets of \(G_n\). The relationship $R_{n,n+1}$ contains all possible combinations of $(S_n,S_{n+1})$ that satisfy confliction constraints imposed by edges connecting \(G_n\) and \(G_{n+1}\). That is to say, any two adjacent elements in a sequence representation must belongs to $R_{n,n+1}$. However, belonging to $R_{n,n+1}$ is not the sufficient condition for a pair $(S_n,S_{n+1})$ to be legal. In fact, apart from conflictions between \(G_n\) and \(G_{n+1}\), there will be constraints on $(S_{n},S_{n+1})$. This is particularly important in generalizing one-dimensional solution to a two-dimensional network, because conflictions may be introduced from the other dimension. So we need to formulate extra constraints, which are written as
\begin{equation}\label{26}
S_n=\alpha_n^{k_n}\in Y_n, Y_n\subset X_n
\end{equation}
This means that for each $\alpha_n^{k_n}$, $k_n$ can only take values in some certain subset of $\{1,2,...,K_n\}$ due to extra constraints.\\

Based on the above definitions, we give out a dynamic programming Algorithm 2 to solve the MWIS problem in a Vector Graph. In this algorithm we use the sequence \(\wp^{k_n}_{(n)}=(S_1S_2...S_n)\) to represent the searching branches of the sequence representation of the best independent set up to step $n$. In fact, $\wp^{k_n}_{(n)}$ is an n-sequence, i.e. an independent set of the first $n$ cells including $G_{1}$ to $G_{n}$, with the assumption that \(S_n\) equals to a specific \(\alpha_n^{k_n}\). In another word, the n-sequence $\wp^{k_n}_{(n)}$ should be written as \((\ast  \ast  \ast \alpha_n^{k_n})\). \(k_n\) obviously denotes the current state in the $n$th step. For each $k_n$, we only reserve one optimal path \(\wp^{k_n}_{(n)}\), which is similar to the classic Viterbi Decoding[12]. By definition, \(|\wp^{k_n}_{(n)}|_N\) still denotes the NWC of $\wp^{k_n}_{(n)}$, which is going to be optimized. Since Algorithm 2 is a direct application of dynamic programming and the proof is quite straightforward, we omit the proof in this paper.\\
\section{Proof of Lemma~\ref{fl_div_lmm}: A floor division scheme}\label{floor_division_scheme}
Now we prove Lemma~\ref{fl_div_lmm}, which indicates that for any $M$ and $L<M$, there is a floor division scheme that guarantees the properties i) to iii). We prove this lemma by explicitly constructing $L$ floor divisions $F=\mathop \cup \limits_{j =  1 }^{Q}  {F_t^j}$, $t$ from 1 to $L$. This construction is also useful in the Algorithm 1. Assume $r=M-L(Q-1)$. We know that $0<r\le L$. For $t=0$, we use~\eqref{11} to build each $F_0^j, \forall j$. We set the marginal rows as $m(0,j)=L(j-1)+1$. For $1<t\le r-1$, we build
\begin{equation}\label{27}
F_t^j=t+F_0^j=\{m\in F|m-t\in F_0^j\},t=1,2,...,r-1
\end{equation}
This is like $t=1$ in Fig. 3, i.e. the second floor division where $r=2$. For these floor divisions, we set $m(t,j)=t+L(j-1)+1$. Note that here + and - are in the sense of modulo $M$. If $r$ equals to $L$, which means that $M$ is divisible by $L$, we have finished building floors. Otherwise, for $t$ from $r$ to $L-1$, we set
\begin{equation}\label{28}
F_t^j=t+F_0^j=\{m\in F|m-t\in F_0^j\},j=1,...,Q-2,\forall t
\end{equation}
\begin{equation}\label{29}
F_t^{Q-1}=\{L(Q-2)+1+t,...,M\},t=r,...,L-1
\end{equation}
\begin{equation}\label{30}
F_t^{Q}=\{1,...,t\},t=r,...,L-1
\end{equation}
For $j$ from 1 to $Q-1$, we still set the marginal rows as $m(t,j)=t+L(j-1)+1$. For $j=Q$, we do not set any rows to be marginal rows. These floor divisions are like the third and fourth divisions in Fig. 3. We clearly see from Fig. 3 that this floor division scheme results in the cyclic behavior of marginal rows, and thus, each element from $\{1,...,M\}$ shows up as the marginal row once. The property (i) in Lemma~\ref{fl_div_lmm} can be easily checked.

\end{document}